\documentclass[runningheads]{llncs}
\usepackage[T1]{fontenc}
\usepackage{graphicx}
\usepackage{eurosym}
\usepackage{graphicx,amsmath,amsfonts}
\usepackage{setspace}
\usepackage{amssymb}
\usepackage{dsfont}
\usepackage{amsmath}
\usepackage{amsfonts}
\usepackage{ifthen}
\usepackage{mathtools}
\usepackage{graphicx}
\usepackage{enumerate}
\usepackage{enumitem}
\usepackage{color}
\usepackage{framed}
\usepackage[margin=1in]{geometry}
\usepackage[dvipsnames]{xcolor}
\usepackage[colorlinks=true, urlcolor=Mahogany, linkcolor=Mahogany, citecolor=Mahogany]{hyperref}
\usepackage{fp}
\usepackage[abspage,user,savepos]{zref}
\usepackage[capitalize]{cleveref}

\usepackage{amsthm}
\usepackage{caption}
\usepackage{subcaption}
\usepackage{palatino}
\usepackage{cuted,balance}
\usepackage[ruled]{algorithm2e} 
\usepackage{bbm}

\DeclareMathOperator*{\argmax}{arg\,max}

\newtheorem{assumption}{Assumption}

\allowdisplaybreaks

\let\originaleqref\eqref
\renewcommand{\eqref}{\originaleqref}
\onehalfspace
\newcommand{\vmax}{v_{\max}}
\newcommand{\vmin}{v_{\min}}
\newcommand{\dd}{\text{d}}
\newcommand{\indic}[1]{\mathbbm{1}[#1]}
\newcommand{\comment}[1]{}

\begin{document}
\title{Deterministic Refund Mechanisms}
%
%
\author{Saeed Alaei\inst{1} \and
Shuchi Chawla\inst{2} \and
Zhiyi Huang\inst{2} \and
Ali Makhdoumi\inst{3} \and
Azarakhsh Malekian\inst{4} }

\authorrunning{Alaei et al.}
%
\institute{Google Research \and
University of Texas at Austin \and
Duke University \and
University of Toronto}
\maketitle              
\begin{abstract}
We consider a mechanism design setting with a single item and a single buyer who is uncertain about the value of the item. Both the buyer and the seller have a common model for the buyer's value, but the buyer discovers her true value only upon receiving the item. Mechanisms in this setting can be interpreted as randomized refund mechanisms, which allocate the item at some price and then offer a (partial and/or randomized) refund to the buyer in exchange for the item if the buyer is unsatisfied with her purchase. Motivated by their practical importance, we study the design of optimal deterministic mechanisms in this setting. We characterize optimal mechanisms as virtual value maximizers for both continuous and discrete type settings. We then use this characterization, along with bounds on the menu size complexity, to develop efficient algorithms for finding optimal and near-optimal deterministic mechanisms.

\keywords{Return policies \and Deterministic mechanism design \and Online retailers \and Approximate optimality.}
\end{abstract}

\section{Introduction}

Buying personal items like clothing or shoes online can be fraught with risk -- without the opportunity to see and try the item before purchasing, the buyer cannot be sure whether the item is worth the price they paid. Online retail platforms like Amazon ease the buyers' dilemma and encourage spending by sharing their risk and offering flexible and generous return options. Returns and refunds are especially important when it is difficult for the platform to provide all of the information the buyer may need to assess whether the item satisfies their needs. For example, clothing is typically marked with one of a handful of size options, but this classification is coarse and may not fit all body types equally well. Likewise, while online retailers often provide pictures of the items they sell, these are not always true to color. Other traits like sturdiness cannot be judged through pictures at all. There is also variation across buyers in terms of how much risk or variation they are willing to undertake. Some buyers are very particular in their needs, whereas others would happily accept a range of qualities and parameters. 

Return policies can help bring more buyers and encourage more spending on a retail platform. But they can also potentially hurt the platform's revenue by allowing buyers access to extra information that they would otherwise not have. Given this tradeoff, as well as uncertainty in buyer needs and item attributes, how should a platform design a practical and near optimal return/refund policy?



This paper examines the design of return and refund schemes as a mechanism design problem where the goal is to maximize the seller's expected revenue. We consider a stylized model for the sale of a single item to a single buyer drawn from a heterogeneous population. The item has attributes that are drawn from a known distribution. The value of the buyer for the item depends on both her type as well as the item's attributes. However, neither the seller nor the buyer can directly observe the attributes before the sale happens. The selling mechanism, therefore, has two phases. The first phase happens after the buyer discovers her type but before she observes her true value for the item. At this phase, the seller allocates the item to the buyer with some probability at some price. The second phase happens after the buyer has seen the item and discovered her true value. At this point, the seller may offer the buyer the opportunity to return the item with some probability and with some additional transfer. 
We show that mechanisms in this model can be interpreted as menus where each option corresponds to an initial price paired with a random refund, which the buyer can obtain upon returning the item. 

Our work focuses on the {\bf computational complexity of optimizing revenue over deterministic refund schemes}. Although optimal mechanisms in general may use randomness, such randomized refund schemes are rarely, if ever, seen in practice. In many mechanism design problems, optimizing over deterministic mechanisms is harder than optimizing over randomized mechanisms because the latter is general a convex set of possible solutions whereas the former is not. Nevertheless, for refund mechanisms we show that, under suitable assumptions, one can efficiently optimize over deterministic mechanisms. 

At a technical level our approach is based on two components. First, we show how to express the mechanism design problem in terms of suitably defined virtual values by deriving a payment identity. Then we provide menu size complexity results for the optimal mechanism in various settings of interest. These combined together yield efficient dynamic programming based algorithms for finding the optimal deterministic refund scheme.


\paragraph{Ordered types.} We now describe the key assumption that drives our results. Each buyer in the market has a type that specifies the distribution of her (a priori unknown) value for the item. We assume that these types are ordered in that a buyer of higher type has a value distribution that first-order stochastically dominates the value distribution at any lower type.
One example of such an "ordered types" setting is when the item has a single ordered attribute (a.k.a. its quality) and the value of the buyer is a weakly increasing function of both her type and the item's quality. So, a better quality item brings (weakly) higher value to the buyer; and a buyer of higher type values the item at any given quality at a (weakly) higher level than a lower type buyer. 


\paragraph{An interdimensional model.} Our ordered types assumption places our model somewhere in between single dimensional and multi dimensional value models in mechanism design. There has been much work recently on so-called "interdimensional" settings in algorithmic mechanism design. Our model most closely resembles the ordered items setting in \cite{chawla2022pricing} and the second example above closely mimics the setting of the FedEx problem of \cite{fiat2016fedex}. However, there are important differences. In our setting, the item type is drawn exogenously and the seller has no control over it, whereas in these prior works the seller controls which item type (i.e., quality) the buyer receives. Thus, incentive constraints in our setting are quite different from those in previous work on interdimensional settings.

\paragraph{Connections with information design.} Our work is also related to information design settings where the seller has information about the item that informs the buyer's value but is not available directly to the buyer such as \cite{daskalakis2016does}. In these settings, the seller can strategically decide how much information to share with the buyer prior to the sale of the item. Refunds can form an important component of mechanisms in this space, potentially simplifying and improving the design of information sharing schemes. In particular, the seller can couple sharing partial information about the item prior to the sale with the offer of a partial refund after the sale. This connection between information design and refunds has not been explored adequately, and we expect that our analysis can guide such connections. 

We discuss connections to prior work in more detail in Section \ref{sec:Related}.

\paragraph{Menu size.} A key driver of our efficient algorithms is the observation that the menu size of the optimal deterministic refund scheme is small. Note that the effective menu size of any mechanism is at most the number of distinct buyer types, $m$ -- menu options that are not bought by any buyer can simply be dropped. However, when $m$ is large, we can nevertheless ensure a small menu size in many contexts. In particular, we consider three discrete settings. In the first, the buyers' values are drawn from a small discrete set of size $n$; we show that in this case the optimal mechanism offers at most $n$ options. In our second model, we assume that the number of values is large but the item is drawn uniformly from $k$ possible types, where $k$ is a small number.\footnote{One example of this setting is a t-shirt with $5$ possible sizes -- XS, S, M, L, and XL. Another example is a shoe that comes in half-integer sizes ranging from 7-13.} In this case, we show that the optimal mechanism has at most $k+1$ options. Our final model again assumes that there are only $k$ possible item types (or qualities) for some small integer $k$, but these are ordered, with each buyer having a weakly  higher value for a higher quality item. In this case, even for arbitrary (non-uniform) distributions over item types, the optimal mechanism contains at most $k+1$ menu options.

We note that in a real world setting a seller might want to further limit the number of menu options it offers in order for the mechanism to be simple for buyers to use. We show that the optimal deterministic mechanism with at most $c$ options can be computed in time polynomial in $c$ and other relevant parameters of the problem.











\subsection{Related literature}\label{sec:Related}
Our paper is closely related to the literature on sequential screening studied in \cite{courty2000sequential}, \cite{esHo2007optimal}, \cite{akan2015revenue},  \cite{bergemann2017scope}, and more recently \cite{bergemann2020scope}, and \cite{von2022optimal}. In particular, these papers study a setting where a consumer sequentially learns her value but knows its distribution at the time of contracting. This literature has focused on deriving simple conditions for the optimality of a deterministic mechanism. For the sake of completeness, we provide such conditions for our formulation as well. We then depart from this literature by asking a complementary question: In the class of deterministic mechanisms, how can we find the optimal one? 


Our paper also relates to the literature on third-degree price discrimination, such as \cite{schmalensee1981output} and more recently \cite{aguirre2010monopoly}, \cite{krahmer2015optimal}, \cite{bergemann2015limits}, \cite{cummings2020algorithmic}, and \cite{bergemann2022third} that analyze the limits of price discrimination and the extent to which a seller can gain by price discrimination. In our setting, by designing a refund mechanism and allowing the buyer to try the item, the seller enables a form of price discrimination: the buyer learns the item value privately, and if the seller designs the mechanism properly, this information can be used to exercise a form of price discrimination. In contrast to the above papers, however, in our setting, both the seller and buyer initially have the same information about the underlying values, and then only the buyer privately gains more information. 


Moreover, our paper relates to studies exploring return policies, such as those by  \cite{inderst2013sales}, \cite{zhang2013revenue}, and \cite{escobari2014price}. The research most closely related to ours is the works by \cite{che1996customer} and \cite{hinnosaar2020robust}. Specifically, \cite{che1996customer} devises a model to scrutinize the optimality of allowing full refunds and establish that they are optimal if consumers are sufficiently risk-averse or if retail costs are high. \cite{hinnosaar2020robust}, on the other hand, contemplates the formulation of a "robust" mechanism (in accordance with \cite{bergemann2005robust}) for a consumer with binary value and finds the mechanism that works well against the worst-case distribution of consumer values. These papers differ in both modeling and analysis but are complementary to our emphasis on characterizing the (approximately) optimal deterministic return policy.


Our paper also relates to the literature on interdimensional mechanism design --- a multiparameter setting where one
parameter is the value of the buyer and others capture some side information that impacts the buyer's value. A leading
 example of this setting is the FedEx Problem of \cite{fiat2016fedex}. Other examples of this are when one of the parameters specifies the buyer’s budget \cite{devanur2017optimal}, \cite{devanur2017optimalb} and when the buyer obtains value from a specific set of items (single-minded buyer) \cite{devanur2020optimal}. Both our formulation and results are different from these papers. In particular, unlike the above papers, the buyer in our setting does not initially have full knowledge about the parameters and gains extra information through the mechanism that the seller designs.\footnote{Less directly related is a line of work that aims to characterize optimal mechanisms beyond single dimensional settings \cite{chawla2010multi}, \cite{briest2010pricing}, \cite{giannakopoulos2014duality}, \cite{daskalakis2013mechanism}, \cite{haghpanah2015reverse}, \cite{malakhov2009optimal}, and \cite{babaioff2017menu}, among others.}

Finally, our paper relates to the literature at the intersection of information design/data acquisition and mechanism design. The settings in which the buyer can acquire more information about her valuation (e.g., at a cost) have been studied in \cite{cremer2009auctions},  \cite{armstrong2016search}, \cite{golrezaei2017auctions}, \cite{alaei2021revenue}, and \cite{lyu2021optimal}, among others. In particular, \cite{lyu2021optimal} study the optimal refund mechanism when an uninformed buyer can privately acquire information about his valuation over time. Settings in which the seller can gain more information about the buyer value and then (strategically) reveals it to the buyer have been studied in   \cite{ben2014optimal}, \cite{li2017discriminatory}, and \cite{li2020mechanism} among others (see also \cite{bergemann2005information} for a survey). Our work is different from both of the above settings as the buyer in our setting has the ability to gain more information, but the design choice of the seller impacts this ability. 

\section{Problem Formulation}\label{sec:formulation}
\paragraph{Buyer types and values.} We consider a multi-dimensional mechanism design setting with a single buyer and a single item. The buyer's value for the item depends on her own type as well as the (random and a priori unknown) attributes of the item. Since buyer types can be renamed, we assume without loss of generality that types lie in $[0,1]$, and are drawn from the uniform 
distribution over a (suitably sized) subset of this range. We use $t$ and $v$ to denote the buyer's type and value respectively, and $T\subseteq [0,1]$ and $V\subseteq \mathbb{R}^+$, respectively, to denote the sets of all possible item types and values. When these sets are discrete, we denote their sizes by $m:=|T|$ and $n:=|V|$. The item's attributes are collectively called its type and denoted by $s$. 

Once the buyer's type $t$ is realized, her value $v$ for the item is distributed according to CDF $G_t$ with density function $g_t$. 
Note that the conditional distribution $G_t$ encodes the composition of the distribution of item attributes and the function mapping the buyer's type $t$ and the item attributes $s$ to the buyer's value. In most of the paper we suppress the details of this composition because the buyer's behavior and the seller's mechanism design problem depend only on the buyer's conditional value distribution.

\paragraph{Ordered buyer types.} To make our model tractable, we assume that the buyer's type space is ordered in that buyers with higher types have higher values compared to buyers with lower types. One way to formalize this is to require that a buyer with type $t$ values {\em every} item type higher than a buyer with type $t'<t$. This approach is somewhat restrictive because, for example, it disallows settings where buyers of different types prefer different item types.\footnote{Suppose, for example, that the item has two types, Red and Blue, distributed uniformly. Suppose that a buyer of type $t_1$ has a value $v_1>0$ for Red and $0$ for Blue, whereas a buyer of type $t_2>t_1$ has a value of $0$ for Red and $v_2$ for Blue. Then, in this example $t_2$ does not value each of Red and Blue higher than $t_1$ but still satisfies our ordered types assumption as long as $v_2>v_1$.} We impose a weaker condition, namely that the value distribution conditioned on $t$ dominates in the first order sense the value distribution conditioned on $t'$. Formally:

\begin{assumption}\label{assump:valuefunction} {\bf (Ordered types)}
For all pairs of buyer types $t, t'\in T$ with $t'<t$, $G_{t}$ first order stochastically dominates $G_{t'}$: for all $v\in V$, $G_t(v)\le G_{t'}(v)$. 
\end{assumption}

\paragraph{Settings of interest.}
We will consider two settings that differ in how the buyer's type and the item's attributes are distributed.
\begin{itemize}
    \item {\bf Continuous differentiable setting.}\label{model:continuous} In this setting, $T=[0,1]$. For all $t\in T$, the value distribution $G_t$ is continuous and differentiable over $\mathbb{R}^+$, that is, the density function $g_t$ is bounded. Finally, for every $v\in V$, $G_t(v)$ is differentiable with respect to $t$, that is, $\frac{\dd G_t(v)}{\dd t}$ is well defined.
    \item {\bf Discrete setting.}\label{model:discrete} In this setting, there is a finite number of discrete buyer types and possible values, i.e., $m,n<\infty$. For every $t\in T$, the distribution $G_t$ is an arbitrary discrete distribution over $V$.
\end{itemize}

\subsection{Return policy and refund mechanisms}\label{sec:returnpolicy}
\paragraph{\bf The timeline of information acquisition and the format of a generic mechanism.} We assume that both the buyer and the seller have common knowledge of the distributions $G_t$. The seller cannot directly observe the buyer or item type. The buyer observes her type $t$ at the beginning of the mechanism, but can only observe the item type (and consequently, her value) if and when she receives the item from the seller. We will focus on the class of (dominant strategy) incentive compatible mechanisms; following the revelation principle, we can restrict our attention to direct revelation mechanisms. 

A generic mechanism in this setting proceeds as follows. First, the buyer and seller learn the common information, $\{G_t(\cdot)\}_{t\in T}$. The seller selects and announces a mechanism. In the first phase, the buyer learns her instantiated type $t$. If she decides to participate, she reports a type $t'$ to the seller. The seller allocates the item with probability $x(t')$ and receives a payment $p(t')$ from the buyer. If the buyer does not receive the item, the mechanism terminates. Otherwise, in the second phase, the buyer observes her value $v\sim G_t$ and reports a value $v'$ to the seller. The buyer retains the item with probability $\bar{z}(t',v')$ and the seller makes a transfer of $\bar{r}(t',v')$ to the buyer. The buyer receives a net allocation probability of $\bar{z}(t',v')x(t')$; and makes a net payment of $p(t')-x(t')\bar{r}(t',v')$ to the seller.

Observe that we can simplify the space of mechanisms to those that always allocate the item to the buyer in the first phase, without losing generality. In particular, given the functions $(p, x, \bar{z}, \bar{r})$, let us define for all $v$ and $t$, $z(t,v):=\bar{z}(t,v)x(t)$ and $r(t,v):=\bar{r}(t,v)x(t)$, and let $x^*$ denote the function that maps all types $t$ to $1$. Then, we note that the mechanisms $(p,x,\bar{z},\bar{r})$ and  $(p, x^*, z, r)$ are equivalent in their final outcomes for {\em every} instantiation of buyer types, item types, and buyer reports. Henceforth, we will drop the argument $x$ and describe a mechanism as a tuple  $(p, z, r)$. 

\begin{definition}[Return policy mechanism]
    A return policy mechanism is defined by three mappings $t \mapsto p(t)\in \mathbb{R}^+\cup\{0\}$, $(t,v) \mapsto r(t,v) \in \mathbb{R}^+\cup\{0\}$, and $(t,v) \mapsto z(t,v) \in [0,1]$ with the buyer-seller interaction defined as in the timeline above. We call the second phase allocation and transfer tuple, $\{z(t,v), r(t,v)\}_{t\in T, v\in V}$, the mechanism's {\em return policy}.

\end{definition}

\paragraph{\bf Refund mechanisms.} Since the focus of this work is on deterministic mechanisms, we will now present a simpler formulation that is more germane to the deterministic setting. In a {\em refund mechanism}, the seller's return policy is to simply offer the buyer a (randomized) refund in exchange for the item. 

\begin{definition}[Refund mechanism]
    A refund mechanism is defined by a menu $(\{p_i, R_i\})$, where the $i$th option consists of real-valued pair of a price $p_i$ and a random refund $R_i$. In the first phase, the buyer chooses a menu option $\{p_i, R_i\}$, makes the payment $p_i$ to the seller, and acquires the item. In the second phase, the seller instantiates $R_i$ and offers this refund to the buyer. The buyer returns the item and accepts the refund if her value is smaller than $R_i$ and otherwise keeps the item with no additional transfers. 
\end{definition}

We prove the following equivalence theorem in the appendix.

\begin{theorem}
\label{thm:return-refund-equiv}
Any return policy mechanism can be implemented as a (randomized) refund mechanism and vice versa.
\end{theorem}

\paragraph{\bf Incentive compatibility.} We use the notion of subgame perfect equilibrium to define incentive compatibility. We assume that the buyer is a quasilinear utility maximizer. Following backward induction, let us first analyze the buyer's utility maximization problem in Phase 2 of the mechanism. Let $t$ and $v$ denote the buyer's true type and value respectively, and suppose the buyer reported type $t'$ in the first phase. Then, in the second phase, reporting value $v'$ brings the buyer a Phase 2 utility of $r(t',v')+v\cdot z(t',v')$. Observe that we did not include the Phase 1 payment $p(t')$ in this expression because the buyer's action in the second phase (of reporting her value) does not affect her Phase 1 payment. The buyer maximizes her Phase 2 utility by reporting 
\begin{align*}
    v^*(t,t',v) := \argmax_{v'}\{ r(t',v')+v\cdot z(t',v')\}.
\end{align*}
Incentive compatibility in Phase 2 requires that if the buyer truthfully reports her type in Phase 1, her Phase 2 utility is maximized by truthfully reporting her value. That is, $v^*(t,t',v) = v$ for all $v,t,t'$. We now turn to Phase 1 incentive compatibility. We will use $Q(t,t')$ to denote the expected utility the buyer receives in Phase 2 if her true type is $t$ and she reports $t'$. In particular,
\begin{align}
    Q(t,t') := \mathbb{E}_{v\sim G_t}\left[\max_{v'}\left\{r(t',v')+  v\cdot z(t',v')\right\}\right] = \mathbb{E}_{v\sim G_t}\left[r(t',v)+  v\cdot z(t',v)\right]. \label{eq:Def-Q}
\end{align}
We can then write the expected utility of the buyer in Phase 1 as $-p(t')+Q(t, t')$. Incentive compatibility requires that this expected utility is maximized at $t'=t$.

\paragraph{\bf The seller's revenue maximization problem.} 
The seller's revenue maximization problem can now be written as the following program that we call \eqref{eq:formulation}.

\begin{align}
    \max_{p(\cdot), r(\cdot, \cdot), z(\cdot, \cdot)} & \mathbb{E}_{(t,v)}\left[p(t) - r(t,v)  \right] \label{eq:formulation}\tag{Opt-Rev}\\
    \text{ s.t. } & -p(t) + Q(t,t) \ge -p(t') + Q(t,t') \text{ for all } t,t' \notag\\
    & -p(t) + Q(t,t) \ge 0 \text{ for all } t \notag \\
    & r(t,v)+ z(t,v) v \ge r(t,v')+ z(t,v') v \text{ for all } t, v,v' \notag
    \\
    & z(t,v) \in [0,1]\, \text{ for all } t, v,  \notag
\end{align}
where the objective is the expected upfront payment minus the expected refunds, the first constraint is the incentive compatibility in the first phase, ensuring that the buyer with type $t$ reports truthfully. The second constraint is the individual rationality ensuring that the buyer will participate in the mechanism, and the third constraint is the incentive compatibility in the second phase, ensuring that the buyer reports the item type $s$ truthfully. Finally, the last constraint guarantees that the allocation probabilities are valid.

\subsection{Deterministic refund mechanisms and incentive compatibility}

In this paper, we are interested in {\em deterministic} refund mechanisms, where the refunds $R_i$ in every menu option are deterministic. Given a menu $(\{p_i, R_i\})$ and a buyer type $t$, we use $p(t)$ and $R(t)$ to denote the price and refund associated with the menu option that the buyer of type $t$ selects. The deterministic refund mechanism then implements the allocation rule $z$ where $z(t,v) = \indic{v\ge R(t)}$; in other words, the buyer of type $t$ purchases the option $(p(t), R(t))$ at a price of $p(t)$ and then keeps the item if and only if her instantiated value exceeds the refund $R(t)$. 

On the other hand, given a $\{0,1\}$-allocation $z(t,\cdot)$ satisfying certain incentive constraints, we can determine a deterministic refund menu implementing this allocation function through appropriate payment identities. In particular, for any type $t$, the $\{0,1\}$-allocation $z(t,\cdot)$ can be implemented by posting a refund of $R(t)$ defined as follows:
\begin{align*}
            R(t) \triangleq  \begin{cases}
            \min\{v : z(t,v)\ge 0\} \quad & \text{if } z(t,1)=1 \\
            v_{\max}  \quad & \text{if } z(t, 1)= 0.\\ 
            \end{cases}
        \end{align*}
The payment $p(t)$ can be determined from the payment identities derived in the appendix and restated below for the continuous and discrete setting respectively. Here, $Q(t,t') := \mathbb{E}_{v\sim G_t}[\max\{v, R(t')\}]$ and $Q_{t'}(\cdot, \cdot)$ is the derivative of $Q$ with respect to its second argument.
\begin{align}
        p(t) & = p(0) + \int_{0}^t Q_{t'} (\tau, \tau) d \tau.& \forall t \in T \label{eq:PI-p}\tag{PI-p-cont} \\
            p(t_i) & =p(t_1)+\sum_{x=2}^{m} (Q(t_{x},t_{x})-Q(t_{x},t_{x-1}))& \forall i \in [m] \label{eq:PI-p-disc}\tag{PI-p-disc}
\end{align}

\noindent
The following theorem establishes a simple monotonicity condition on IC deterministic mechanisms.

\begin{theorem}
\label{thm:det-ic}
    The following statements for a deterministic mechanism $(p, R)$ with allocation function $z$ are equivalent:
    \begin{enumerate}
        \item The mechanism satisfies IC.
        \item $z(t,v)\in \{0,1\} \text{ is weakly increasing in } t \text{ and } v$.
        \item $R(t)$ weakly decreases in $t$ and $p(t)$ satisfies the identity \eqref{eq:PI-p} or \eqref{eq:PI-p-disc}.
    \end{enumerate}
\end{theorem}
\noindent
We observe that $Q(t,t')$ weakly decreases in $t'$ as $R(t')$ weakly decreases in $t'$. Therefore, the payment identities established above imply that the payment function $p$ also decreases weakly in the type of the buyer. In other words, buyers with higher types purchase options with a lower upfront price {\em as well as} a lower refund.

We discuss the complete formulation of IC constraints for return policy mechanisms in the appendix. 

\subsection{A virtual value formulation}
\label{sec:charac:virtualvalue}

The seller's problem  \eqref{eq:formulation} is not a linear optimization problem because $Q(t,t')$ is a nonlinear function of the variables $r(t,v)$ and $z(t,v)$ (because of the max in the expectation). Nonetheless, the seller's problem can be rewritten as a virtual welfare maximization problem with a novel virtual value function. Furthermore, Assumption~\ref{assump:valuefunction} allows us to simplify the IC constraints to obtain a linear formulation of the problem. This in turn drives the algorithmic results stated below. More detailed discussions and proofs are written in the appendix.

\begin{definition}[Virtual value]\label{def:virtual:main}
In the continuous differentiable setting, given the buyer's conditional value distributions $\{G_t\}_{t\in T}$, the virtual value at any pair $(t,v)$ for $t\in T$ and $v\in V$ is defined as follows, where we recall that $F$ is the Unif$[0,1]$ distribution.

    \begin{align}
        \phi(t,v)= v + \frac{1-F(t)}{f(t)} \frac{1}{g_t(v)} \frac{\dd G_t(v)}{\dd t} = v + (1-t) \frac{1}{g_t(v)} \frac{\dd G_t(v)}{\dd t}.
        \label{eq:vv-cont}
        \tag{VV-cont}
    \end{align}
In the discrete setting with the buyer's type set and value set being denoted by $T=\{0=t_1< t_2 <\cdots < t_m=1\}$ and $V=\{v_1 < \cdots < v_n\}$ respectively, the virtual value at $(t_i, v_j)$ for $i\in [m]$ and $j\in [n]$ is given by
\begin{align}
  \phi(t_i,v_j) & = v_j +(m-i) \frac{(v_j-v_{j-1})}{g_{t_i}(v_j)} (G_{t_{i+1}}(v_{j-1})-G_{t_{i}}(v_{j-1}))  \label{eq:vv-disc}
        \tag{VV-disc}
    \end{align}
\end{definition}
\begin{lemma}\label{thm:charac:wo:IC}
    The expected revenue of any IC mechanism $(z,p,r)$ is equal to its expected virtual welfare, $\mathbb{E}_{t,v}\left[ z(t,v) \phi(t,v)\right]$, minus the utility of the buyer of lowest type, $p(0)-Q(0,0)$, with the virtual value $\phi$ as defined above in \eqref{eq:vv-cont} or \eqref{eq:vv-disc}. 
\end{lemma}
Observe that virtual values, as defined above, are always smaller than the buyer's true value, as $G_t(v)$ decreases in $t$ and the second term is negative. Henceforth we will assume that the seller sets $p(0)=Q(0,0)$.

\section{Optimal deterministic mechanism in the discrete model}\label{sec:discrete_few}
We now shift our attention to the discrete model, Model \ref{model:discrete}, focusing on deterministic mechanisms. Recall that in this setting, buyer types and values both belong to discrete sets $T$ and $V$, respectively. Let $m := |T|$ and $n := |V|$. We will further discuss two other models under this discrete setting, separately in Sections 
\ref{sec:discrete_ordered} and  \ref{sec:discrete_uniform}.  For ease of notation, we rename buyer types and values so that $T=[m]$ and $V=\{v_1,v_2,\dots,v_n\}$.


\subsection{Menu size}

We first claim that in this setting, any deterministic mechanism need only have at most $n$ distinct options on its menu.

\begin{lemma}\label{thm:menusize_1}
    For the discrete setting (Model \ref{model:discrete}) with $|V|=n$, any deterministic refund mechanism can be implemented with at most $n$ menu options. Furthermore, we may assume without any loss in the mechanism's revenue, that the refund corresponding to the $i$th option is $v_i$ for $i\in[n]$. 
\end{lemma}

Before proving~\cref{thm:menusize_1}, we first describe the main idea behind. Since there are only $n$ different values buyers could take, it divides the range of refund we offered to different types of buyers into $n$ groups, enabling us to also divide buyers into $n$ groups. Within each group, different types of buyers receive identical allocation according to the instantiated value, thus we can combine all the menu options for all types of buyers in this group into a single menu choice. Now we delve into the details of our proof.

\begin{proof}
The starting point is to derive a simple expression for the pricing in this case as a function of the deterministic refunds $R(t)$ defined previously. We first observe that the Phase 1 expected utility functions for deterministic mechanisms can be written as:
\[Q(t,t') = \mathbb{E}_{v\sim G_t}[\max(v,R(t'))].\]
The IC constraints over $t$ and $t'$ with $t>t'$ can then be written as:
\begin{align*}
    Q(t',t)-Q(t',t') & \le p(t)-p(t') \le Q(t,t)-Q(t,t')
\end{align*}
Fixing $R(t)$ for all $t$, we can then maximize the expected revenue by setting $p(t)$ for $t=1, \cdots, m$ in succession according to the tightest upper bound on it. We claim that the tightest upper bound is given by the consecutive IC constraints, that is, by taking $t'=t-1$. In particular, for any fixed choice of $p(1)$, define $p(t)$ for $t>1$ as:
\begin{align}
    p(t)=p(1)+\sum_{y=2}^{t} \{Q(y,y)-Q(y,y-1)\} \label{eq:price-formula}
\end{align}
We claim that this choice of $p(t)$ satisfies all of the IC constraints above. To see this claim, we first observe that the function $Q(t,y)-Q(t,y-1)$ is an increasing function of $t$. This can be seen simply by noting that $\max(v,R(y))-\max(v,R(y-1))$ is an increasing function of $v$ because $R(y)\le R(y-1)$. Then we note that per our definition of $p$, for $t>t'$, $p(t)-p(t') = \sum_{y=t'+1}^{t} \{Q(y,y)-Q(y,y-1)\} \le \sum_{y=t'+1}^{t} \{Q(t,y)-Q(t,y-1)\} = Q(t,t)-Q(t,t')$, thereby proving our claim.

We are now ready to prove the menu size theorem. Consider any deterministic IC mechanism $(\{p(t), R(t)\})$ with payments $p(t)$ as described above. Let us divide buyer types into $n$ sets, $T_1, \cdots, T_n$, as follows. For buyer type $t$, let $i\in [n]$ be the item value for which $v_{i-1}< R(t)\le  v_{i}$. Set $i=1$ if $R(t)<v_1$ and set $i=n$ if $R(t)>v_n$. We place buyer type $t$ in set $T_i$. 
    
If there exist any buyer type $t$ with $R(t)>v_n$, it means that we are not allocating anything to this buyer at every possible buyer values. Then we can set $R(t)=v_n$ which the buyer still receives identical expected utilities after the change. The seller also receives identical expected revenue. Thus we can assume that all menu options have a refund no bigger than $v_n$. Now consider any set $T_i$, if there exists two buyer types $t, t'\in T_i$ with $t>t'$ and that $R(t)\neq R(t')$. We can always find buyer type $t$ such that $t,t+1\in T_i$ where $R(t)\neq R(t+1)$.  Note that $t$ and $t+1$ obtain identical allocations $z(\cdot, v)$. Applying our formula for prices derived above, we note that $p(t+1) = p(t)+(Q(t+1,t+1)-Q(t+1,t))=p(t)-G_{t+1}(v_{i-1})(R(t)-R(t+1))$, where $G_{t+1}(v_{i-1})=\sum_{i'<i} g_{t+1}(v_{i-1})$, and therefore, the buyer with type $t+1$ receive identical expected utilities from the two outcomes while the buyer from type $t$ receive higher utility using menu option $(p(t),R(t))$. And the seller receives identical expected revenue whether he adds a menu option of $(p(t+1),R(t+1))$ or not. Accordingly, we can simply drop the menu option corresponding to $t+1$ and assign to it the same menu option (same allocation, refund, and payment) corresponding to $t$. Buyer preferences and the seller's expected revenue remain the same as before. This implies the first part of the theorem. 
    
    For the second part, consider raising the refund corresponding to the set $T_i$ by some $\epsilon>0$ without affecting its allocation function, and let us recompute prices $p(t)$ corresponding to the new $Q$ functions. Then we note that the price for types in sets $T_{i+1}, \cdots, T_{n}$ remains the same as before. The price for any types t in set $T_i$ increases by precisely $G_{t_i^*}(v_{i-1})\epsilon$ where $t_i^*$ is the smallest type in set $T_i$. While those buyers' refund only increases by $G_t(v_{i-1})$, the seller's obtains a weakly higher expected revenue from buyers in set $T_i$. Now consider the smallest type in set $T_{i-1}$ and call it $y$. We note that the change in the refund for $T_i$ increases $Q(y,y-1)$ by exactly $G_y(v_{i-1})\epsilon$. Coupled with the increase of $G_{t_i^*}(v_{i-1})\epsilon$ in the price for type $y-1$ in set $T_i$, where $G_{t_i^*}(v_{i-1})\ge G_y(v_{i-1})$, this leads to a new price for $T_{i-1}$ that is no smaller than the previous price for this set. Higher types also see a corresponding increase in price. However, types in sets $T_{i-1}$ or lower see no increase in refunds. Consequently, the change leads to a weakly higher expected revenue for the seller.   

\end{proof}    

\subsection{A dynamic program to compute the optimal deterministic mechanism}\label{subsec:dp1}

The characterization in \Cref{thm:menusize_1} allows us to develop a dynamic program to find the optimal deterministic mechanism in discrete settings. Let $\{T_i\}$ denote the sets that buyer types are partitioned into, as in the proof of \Cref{thm:menusize_1}, and $t^*_i$ denote the lowest type in set $T_i$. We will define $n$ menu options $(p_i, R_i)$, each corresponding to the set $T_i$. Our dynamic program computes values $U(\ell, j)$ for all $\ell\in[m]$ and $j\in[n]$, which denotes the optimal revenue we can obtain through types $1,\cdots,\ell$ while placing type $\ell$ in the set $T_j$. This could be done by considering all the possibilities of the set type $\ell-1$ belongs to, and compute the corresponding contribution of revenue from type $\ell$ using the virtual value expression we have in \eqref{eq:vv-disc}. After figuring out the allocation of the optimal deterministic mechanism. We will give a menu option $(p_i,R_i)$ for each set $T_i$ where $i\in[n]$, the refund corresponding to it is given by $R_i=v_i$. Let $j$ denote the smallest index $>i$ for which the set $T_j$ is non-empty. Then, we can compute $p_i = p_j - \sum_{\ell=i}^{j-1}g_{t_i^*}(v_\ell)(v_j-v_\ell)$. The algorithm is formally stated in Algorithm~\ref{alg:DPdiscrete}.

\begin{algorithm*}[t]
	\SetAlgoNoLine
	\KwIn{$g_t(v)$ and $\phi(t,v)$ for all $t\in[m]$ and $v\in\{v_1,\cdots,v_n\}$}
\begin{itemize}
    \item Run the following dynamic programming to find $S(\ell)$(for $\ell\in[m]$) such that $z(\ell,v_j)=1$ iff $j \ge S(\ell)$:
    \begin{itemize}
        \item For $\ell = 1$ to $m$ and $j = 1$ to $n$ do:\\
        $U(\ell,j)=\sum_{j'=j}^n g_{\ell}(v_{j'}) \phi(\ell,v_{j'})+ \max_{\bar{j} \le j} U(\ell-1, \bar{j})$
        \item For $\ell = 1$ to $m$ do:\\
        $S(\ell)=\argmax_{j} U(\ell,j)$
    \end{itemize}
    \item For all $ i\in[n]$, let $R_i=v_i$ and $p_i = p_j - \sum_{\ell=i}^{j-1}g_{t_i^*}(v_\ell)(v_j-v_\ell)$ where $j$ is the smallest index $> i$ for which there exist some $\bar{\ell}\in[m]$ such that $S(\bar{\ell})=j$.
\end{itemize} 
\textbf{Output:} For all $ i\in[n]$, output menu option $(p_i,R_i)$

\caption{An DP for finding the optimal deterministic menu in the discrete setting}
\label{alg:DPdiscrete}
\end{algorithm*}

Note that we only compute values $U(\ell,j)$ for all $\ell\in[m]$ and $j\in[n]$ in our dynamic program, we obtain the following theorem.

\begin{theorem}\label{thm:DPdiscrete1}
    For the discrete setting with $m$ buyer types and $n$ buyer values, there exists a dynamic program that runs in time $O(mn^2)$ and returns the optimal deterministic refund mechanism. This mechanism offers at most $n$ menu options.
\end{theorem}

\subsection{Efficient computation of optimal menus with small size}\label{subsec:dp4}

After settling the computation of the optimal menu, we note that in some circumstances the menu size is quite large ($n$), while we ask if we can still obtain the same revenue using menus with smaller sizes. To answer this question, utilizing similar ideas of dynamic programming, we provide a way to compute the optimal menu with any certain size.

We use $c$ to denote the size of our small menu, where $c<n$.

Since we only have $c$ possible menu options, in contrast to Section~\ref{subsec:dp1}, we can only partition the buyers into $c$ sets, where the buyers inside the same set obtain identical allocation. Let $\{T_1,T_2,\cdots,T_c\}$ denote the sets that buyer types are partitioned into. We not only need to decide how we partitioning different types of buyers, we also need to decide where we start to allocating the item within each set. To be more specific, for $q\in [c]$, let $t_q^*$ denote the lowest type in set $T_q$. Let $i_q$ denote the parameter where $z(t_q^*,v_{i})=0$ if and only if $i<i_q$. This brings extra computation difficulty for our dynamic program, where we have to compute values $U(\ell, j,i_j)$ for all $\ell\in[m]$, $j\in[c]$ and $i_j\in[n]$, which denotes the optimal revenue we can obtain through types $1,\cdots,\ell$ while placing type $\ell$ in the set $T_j$ and making $i_j$ the parameter for set $T_j$.  We will give a menu option $(p_q,T_q)$ for each set $T_q$ where $q\in[c]$, the refund corresponding to it is given by $R_q=v_{i_q}$. Let $j$ denote the smallest index $>q$ for which the set $T_j$ is non-empty. Then, we can compute $p_q = p_j - \sum_{\ell=i_q}^{i_j-1}g_{t_q^*}(v_\ell)(v_{i_j}-v_\ell)$. The algorithm is formally stated in Algorithm~\ref{alg:DPdiscreteSmall}.
\begin{algorithm*}[H]
	\SetAlgoNoLine
	\KwIn{$g_t(v)$ and $\phi(t,v)$ for all $t\in[m]$ and $v\in\{v_1,\cdots,v_n\}$}
\begin{itemize}
    \item Run the following dynamic programming to find $S(\ell)$(for $\ell\in[m]$) and $i_q$(for $q\in[c]$) such that $z(\ell,v_j)=1$ iff $j \ge i_{S(\ell)}$:
    \begin{itemize}
        \item For $\ell = 1$ to $m$, $j = 1$ to $c$ and $x=1$ to $n$ do:\\
        $U(\ell,j,x)=\sum_{x'=x}^n g_{\ell}(v_{x'}) \phi(\ell,v_{x'})+ \max_{\bar{j} \le j, \bar{x} \ge x} U(\ell-1, \bar{j},\bar{x})$
        \item For $\ell = 1$ to $m$ do:\\
        $S(\ell)=\argmax_{j} U(\ell,j,x)$
        \item For $j=1$ to $c$ do: \\
        If there exist some $\bar{\ell}$ such that $S(\bar{\ell})=j$, then $i_j=\argmax_{x}U(\bar{\ell},j,x)$.
    \end{itemize}
    \item For all $ q\in[c]$, let $R_q=v_{i_q}$ and $p_q = p_j - \sum_{\ell=i_q}^{i_j-1}g_{t_q^*}(v_\ell)(v_{i_j}-v_\ell)$ where $j$ is the smallest index $> q$ for which there exist some $\bar{\ell}\in[m]$ such that $S(\bar{\ell})=q$.
\end{itemize} 
\textbf{Output:}  For all $q\in[c]$, output menu option $(p_q,R_q)$.
\caption{An DP for finding the optimal small-sized deterministic menu in the discrete setting}
\label{alg:DPdiscreteSmall}
\end{algorithm*}

Note that we only compute values $U(\ell, j,x)$ for all $\ell\in[m]$, $j\in[c]$ and $x\in[n]$. We obtain the following theorem.

\begin{theorem}
    For the discrete setting with $m$ buyer types and $n$ buyer values, there exists a dynamic program that runs in time $O(mn^2c^2)$ and returns the optimal deterministic refund mechanism with menu size $\le c$.
\end{theorem}



\section{Discrete Setting with ordered item types}\label{sec:discrete_ordered}

Now we discuss another model of the discrete setting, where we place a further assumption. As in the previous setting, $m,n<\infty$. Let $S$ denote the set of item types, with $k:=|S|$. We assume that item types are ordered in that for {\em every} buyer type $t$ and every pair of item types $s>s'\in S$, we have $v(t,s)\ge v(t,s')$. The item types can be also viewed as the quality of the item, where naturally buyers hold larger value towards item with higher quality.

This setting is algorithmically interesting when $k\ll m,n$. For ease of notation, we rename item types so that $S=[k]$.

\subsection{Menu size}

We first claim that in this setting, we can employ a menu with significantly smaller size, where any deterministic mechanism need only have at most $k+1$ distinct options on its menu.

\begin{lemma}\label{thm:discrete}
    For the discrete setting with ordered item types with $|S|=k$, any deterministic refund mechanism can be implemented with at most $k+1$ menu options. Furthermore, we may assume without any loss in the mechanism's revenue, that the refund corresponding to the $0$th option is $v_{\min}$; that corresponding to the $i$th option for $i\le k$ is equal to the value $v(t,i)$ of the smallest buyer type $t$ that purchases this option; and that corresponding to option $k+1$ is $v_{\max}$. 
\end{lemma}

The main idea of the proof of \Cref{thm:discrete} is similar to the proof of \Cref{thm:menusize_1}. In contrast of dividing the range of refund using $n$ values, we divide the range of refund using $k$ item types. Now within each group, different types of buyers receive identical allocation according to the item types, enabling us to combine all the menu options for all types of buyers in this group into a single menu choice. Below are the details of our proof.

\begin{proof}
Similarly, we can derive the same expression for the pricing in this case as a function of the deterministic refunds $R(t)$ defined previously. 

Fixing $R(t)$ for all $t$, applying the same argument, we can obtain the best choice for the payments $p(t)$. In particular, for any fixed choice of $p(1)$, define $p(t)$ for $t>1$ as:
\begin{align}
    p(t)=p(1)+\sum_{y=2}^{t} \{Q(y,y)-Q(y,y-1)\} \label{eq:price-formula}
\end{align}
Again we are able to claim and prove that this choice of $p(t)$ satisfies all of the IC constraints above using the same arguments. 

We are now ready to prove the menu size theorem. The following proofs are similar to the previous proof of \Cref{thm:menusize_1}. Consider any deterministic IC mechanism $(\{p(t), R(t)\})$ with payments $p(t)$ as described above. Let us divide buyer types into $k+1$ sets, $T_0, \cdots, T_k$, as follows. For buyer type $t$, let $i\in [k]$ be the item quality for which $v(t,i)\le R(t)<  v(t,i+1)$. Set $i=0$ if $R(t)>v(t,1)$. We place buyer type $t$ in set $T_i$. 
    
Now consider any two buyer types $t, t'\in T_i$ with $t>t'$ and suppose that $R(t)\neq R(t')$. Note that $t$ and $t'$ obtain identical allocations $z(\cdot, s)$. Applying our formula for prices derived above, we note that $p(t) = p(t')-G(i)(R(t')-R(t))$, where $G(i)=\sum_{i'<i} g(i)$, and therefore, the two buyer types receive identical expected utilities from the two outcomes and the seller receives identical expected revenue from the two outcomes. Accordingly, we can simply drop the menu option corresponding to $t'$ and assign to it the same menu option (same allocation, refund, and payment) corresponding to $t$. Buyer preferences and the seller's expected revenue remain the same as before. This implies the first part of the theorem. 
    
    For the second part, consider raising the refund corresponding to the set $T_i$ by some $\epsilon>0$ without affecting its allocation function, and let us recompute prices $p(t)$ corresponding to the new $Q$ functions. Then we note that the price for types in sets $T_1, \cdots, T_{i-1}$ remains the same as before. The price for types in set $T_i$ increases by precisely $G(i)\epsilon$, countering the respective increase in refunds and thereby causing no effect on the seller's expected revenue. Now consider the smallest type in set $T_{i+1}$ and call it $y$. We note that the change in the refund for $T_i$ increases $Q(y,y-1)$ by at most $G(i')\epsilon$ for some $i'\le i$. Coupled with the increase of $G(i)\epsilon$ in the price for set $T_i$, where $G(i)\ge G(i')$, this leads to a new price for $T_{i+1}$ that is no smaller than the previous price for this set. Higher types also see a corresponding increase in price. However, types in sets $T_{i+1}$ or higher see no increase in refunds. Consequently, the change leads to a weakly higher expected revenue for the seller.   

\end{proof}

\subsection{A dynamic program to compute the optimal deterministic mechanism}\label{subsec:dp3}

Again, the characterization in \Cref{thm:discrete} allows us to develop a dynamic program to find the optimal deterministic mechanism in this setting. We employ similar ideas as \cref{subsec:dp1}, we again let $\{T_i\}$ denote the sets that buyer types are partitioned into, as in the proof of \Cref{thm:discrete}, and $t^*_i$ denote the lowest type in set $T_i$. We will now define $k+1$ menu options $(p_i, R_i)$, each corresponding to the set $T_i$. Still, our dynamic program computes values $U(\ell, j)$ for all $\ell\in[m]$ and $j\in[k+1]$, which denotes the optimal revenue we can obtain through types $1,\cdots,\ell$ while placing type $\ell$ in the set $T_j$. This will be done in a similar way where we consider all the possibilities of placement and compute its contribution to our revenue using virtual values. After figuring out the optimal allocation using the dynamic program, we employ similar ideas to compute the price $p_i$ and refund $R_i$ for each set $T_i$ as one menu option. The formal algorithm is stated in Algorithm~\ref{alg:DPdiscrete_ordered}.

\begin{algorithm*}[H]
	\SetAlgoNoLine
	\KwIn{$g_t(v)$ and $\phi(t,v)$ for all $t\in[m]$ and $v\in\{v_1,\cdots,v_n\}$}
\begin{itemize}
    \item For $i=k+1$ do:\\
    $p_{k+1}=R_{k+1}=v_{\max}$
    \item Run the following dynamic programming to find $S(\ell)$(for $\ell\in[m]$) such that $z(\ell,v)=1$ iff $v\ge v(\ell,S(\ell))$:
    \begin{itemize}
        \item For $\ell = 1$ to $m$ and $j = 1$ to $k$ do:\\
        $U(\ell,j)=\sum_{j'=j}^k g_{\ell}(v(\ell,j')) \phi(\ell,v(\ell,j'))+ \max_{\bar{j} \le j} U(\ell-1, \bar{j})$
        \item For $\ell = 1$ to $m$ do:\\
        $S(\ell)=\argmax_{j} U(\ell,j)$
    \end{itemize}
    \item For all $ i\in[k]$, assume $t^*_i$ denote the lowest type in set $T_i$. Let $R_i=v(t_i^*,i)$ and $p_i = p_j - \sum_{s=1}^{i}g(s)(v(t_j^*,j)-\max(v(t_i^*,s),v(t_i^*,i)))$ where $j$ is the smallest index $> i$ for which there exist some $\bar{\ell}\in[m]$ such that $S(\bar{\ell})=j$.
\end{itemize} 
\textbf{Output:} For all $ i\in[k+1]$, output menu option $(p_i,R_i)$

\caption{An DP for finding the optimal deterministic menu in the discrete setting with ordered item types}
\label{alg:DPdiscrete_ordered}
\end{algorithm*}

Note that we only compute values $U(\ell,j)$ for all $\ell\in[m]$ and $j\in[k]$ in our dynamic program, we obtain the following theorem.

\begin{theorem}
    For the discrete setting with $m$ buyer types and $k$ ordered item types, there exists a dynamic program that runs in time $O(mk^2)$ and returns the optimal deterministic refund mechanism. This mechanism offers at most $k+1$ menu options.
\end{theorem}

\subsection{Efficient computation of optimal menus with small size}
We now turn our attention again to the setting where we are restricted to use menu size no bigger than $c$ where $c\le k$. Since we only have $c$ possible menu options, we still can only partition the buyers into $c$ sets, where the buyers inside the same set obtain identical allocation(with respect to the item types). Let $\{T_1,T_2,\cdots,T_c\}$ denote the sets that buyer types are partitioned into. Utilizing similar ideas in Section~\ref{subsec:dp4}, now we need to decide how to partition different types of buyers and the item types where we start to allocating the item within each set. For $q\in [c]$, let $t_q^*$ denote the lowest type in set $T_q$. This time we need to decide parameter $y_q$, which denote the parameter where $z(t_q^*,v(t_q^*,y))=0$ if and only if $y<y_q$. Our dynamic program will compute values $U(\ell, j, y_j)$ for all $\ell\in[m]$, $j\in[c]$ and $y_j\in[k]$, which denotes the optimal revenue we can obtain through types $1,\cdots,\ell$ while placing type $\ell$ in the set $T_j$ and making $y_j$ the parameter for set $T_j$.  We will give a menu option $(p_q,T_q)$ for each set $T_q$ where $q\in[c]$, where similar ideas before will be used to calculate $p_q$ and $R_q$. The algorithm is formally stated in Algorithm~\ref{alg:DPdiscreteSmall2}.

\begin{algorithm*}[H]
	\SetAlgoNoLine
	\KwIn{$g_t(v)$ and $\phi(t,v)$ for all $t\in[m]$ and $v\in\{v_1,\cdots,v_n\}$}
\begin{itemize}
    \item Run the following dynamic programming to find $S(\ell)$(for $\ell\in[m]$) and $y_q$(for $q\in[c]$) such that $z(\ell,v)=1$ iff $v \ge v(\ell,y_{S(\ell)})$:
    \begin{itemize}
        \item For $\ell = 1$ to $m$, $j = 1$ to $c$ and $x=1$ to $k$ do:\\
        $U(\ell,j,x)=\sum_{x'=x}^k g_{\ell}(v(\ell,x') \phi(\ell,v(\ell,x'))+ \max_{\bar{j} \le j, \bar{x} \ge x} U(\ell-1, \bar{j},\bar{x})$
        \item For $\ell = 1$ to $m$ do:\\
        $S(\ell)=\argmax_{j} U(\ell,j,x)$
        \item For $j=1$ to $c$ do: \\
        If there exist some $\bar{\ell}$ such that $S(\bar{\ell})=j$, then $y_j=\argmax_{x}U(\bar{\ell},j,x)$.
    \end{itemize}
    \item For all $ q\in[c]$, assume $t^*_q$ denote the lowest type in set $T_q$. Let $R_q=v(t_q^*,s_q)$ and $p_q = p_j - \sum_{s=1}^{s_q}g(s)(v(t_j^*,s_j)-\max(v(t_q^*,s),v(t_q^*,s_q)))$ where $j$ is the smallest index $> i$ for which there exist some $\bar{\ell}\in[m]$ such that $S(\bar{\ell})=j$.
\end{itemize} 
\textbf{Output:}  For all $q\in[c]$, output menu option $(p_q,R_q)$.
\caption{An DP for finding the optimal small-sized deterministic menu in the discrete setting with ordered item types}
\label{alg:DPdiscreteSmall2}
\end{algorithm*}

Note that we only compute values $U(\ell, j,x)$ for all $\ell\in[m]$, $j\in[c]$ and $x\in[k]$. We obtain the following theorem.

\begin{theorem}
    For the discrete setting with $m$ buyer types and $k$ ordered item types, there exists a dynamic program that runs in time $O(mk^2c^2)$ and returns the optimal deterministic refund mechanism with menu size $\le c$.
\end{theorem}


\section{Discrete setting with uniformly distributed item types}\label{sec:discrete_uniform}

Now we discuss the final model of the discrete setting with few item types where the item types are not necessarily ordered. As before, we have $m,n<\infty$, and there are $k$ item types with $k \ll m,n$. We still assume that the buyer types are ordered. In the third setting we assume that the item types are distributed uniformly. Since the buyer's value is a function of the buyer's type and the item's type, the conditional value distribution $G_t$ is uniform over a set of size $k$. Note that the number of distinct values is $n\le mk$. This setting is also algorithmically interesting when $k \ll m,n$. As before, for ease of notation, we rename item types so that $S=[k]$. 

Although item types in this setting are not ordered, because they are uniformly distributed, this setting can be ``reduced'' to the ordered-item-types setting of Section~\ref{sec:discrete_ordered} by permuting buyer values. We show that, as a consequence, the optimal menu size for this setting is still small, in particular at most $k+1$.

Note that in this model, different types of buyers have different preferences towards different item types. However, since the item type is drawn uniformly at random, for buyer of type $t$, there exists a permutation $\sigma_t:[k]\rightarrow[k]$, such that $v(t,\sigma_t(1))\le v(t,\sigma_t(2))\le \dots\le v(t,\sigma_t(k))$. Without loss of generalization, we can assume that for buyer of type $t$, $v'(t,s)=v(t,\sigma_t(s))$ for every $s\in [k]$. By switching from $v$ to $v'$, we incur no change on the revenue of any menu. Thus we can view $v'$ as the true value for each type of buyer, which brings us back to the setting in Section~\ref{sec:discrete_ordered}. 

Applying results in Section~\ref{sec:discrete_ordered}, we obtain the following theorems.

\begin{lemma}\label{thm:discrete_uni}
    For the discrete setting with uniformly distributed item types with $|S|=k$, any deterministic refund mechanism can be implemented with at most $k+1$ menu options. Furthermore, we may assume without any loss in the mechanism's revenue, that the refund corresponding to the $0$th option is $v_{\min}$; that corresponding to the $i$th option for $i\le k$ is equal to the value $v(t,\sigma_t(i))$ of the smallest buyer type $t$ that purchases this option; and that corresponding to option $k+1$ is $v_{\max}$. 
\end{lemma}

\begin{theorem}
    For the discrete setting with $m$ buyer types and $k$ uniformly distributed item types, there exists a dynamic program that runs in time $O(mk^2)$ and returns the optimal deterministic refund mechanism. This mechanism offers at most $k+1$ menu options.
\end{theorem}

\begin{theorem}
    For the discrete setting with $m$ buyer types and $k$ uniformly distributed item types, there exists a dynamic program that runs in time $O(mk^2c^2)$ and returns the optimal deterministic refund mechanism with menu size $\le c$.
\end{theorem}


%
%
%
\bibliographystyle{splncs04}
%
\bibliography{References}

\newpage
\appendix

\section{Proofs and Discussions for Section~\ref{sec:formulation}}\label{sec:proofs}
We first prove Theorem~\ref{thm:return-refund-equiv}.

\vspace{0.1in}
\noindent
{\em Proof of Theorem \ref{thm:return-refund-equiv}.}
Refund mechanisms are by definition a special case of return policy mechanisms. We will now show the other direction. Consider any return policy mechanism $M=(p,z,r)$. We will construct a refund mechanism $M'$, one menu option at a time. For each buyer type $t$, we add to the menu an option $(p_t, R_t)$, where $R_t$ is defined procedurally as follows. We pick a random number uniformly between $0$ and $1$, call it $X$; let $S$ be the smallest value for which $z(t,S)\ge X$; we set $R_t=S$. If no such value $S$ exists, set $R_t=v_{\max}$. Observe that a buyer with type $t$ accepts the refund and returns the item if and only if her true value $v$ is smaller than the refund $R_t=S$. This is the probability that $X=z(t,S)$ is larger than $z(t,v)$. Because $X$ is uniformly distributed over $[0,1]$, this probability is equal to $1-z(t,v)$. We define prices for the menu options using the payment identity \eqref{eq:PI-p}. 

The mechanisms $M$ and $M'$ offer identical allocations to the buyer at all types.  Thus by the payment identity we know that the only remaining thing is to show that the choice of refunds $r(t,\vmax)$ at any type $t$ does not affect the seller's expected revenue. To see this, observe that for the two mechanisms $M=(p,z,r)$ and $\hat{M}=(\hat{p}, z, \hat{r})$, the payment identity \eqref{eq:PI-r} implies that for all $t,v$, $r(t,v)-r(t,\vmax)=\hat{r}(t,v)-\hat{r}(t,\vmax)$. Plugging this into the definition of $Q$ \eqref{eq:Def-Q}, we get that for all $t$ and $t'$, $Q(t,t')-r(t',\vmax) = \hat{Q}(t,t')-\hat{r}(t',\vmax)$. Finally, plugging this into the payment identity \eqref{eq:PI-p}, we get that $p(t)-r(t,\vmax)=\hat{p}(t)-\hat{r}(t,\vmax)$. Putting these together into the expression for revenue, we observe that the revenue is independent of the choice of $r(t,\vmax)$. This shows that mechanisms $M$ and $M'$ obtain same revenue and finished the proof of this theorem. \hfill $\square$

\subsection{Payment identity for the continuous setting}

We begin by establishing a payment identity that expresses the payments $p$ and the refunds $r$ in terms of the allocation function $z$. Recall that the buyer's second phase utility is given by the following expression when her true type is $t$ and she reports $t'$: $Q(t,t') = \mathbb{E}_{v\sim G_t}[\max_{v'}\{r(t', v')+v\cdot z(t', v')\}]$. We use $Q_t$ and $Q_{t'}$ respectively to denote the derivative of this function with respect to its first and second arguments respectively. We likewise use the notation $r_t$, $r_v$, $z_t$, and $z_v$ to denote the partial derivatives of the functions $r$ and $z$ with respect to the arguments $t$ and $v$ respectively.

\begin{definition}[Payment Identity]\label{Def:PaymentIdentity}
    Let $\vmax$ denote the maximum value in the joint support of the distributions $G_t$ for $t\in T$. For any allocation function $z$, the payment identities for payments $p$ and refunds $r$ are defined as:
    \begin{align}
        r(t, v) & = r(t,\vmax) + \vmax z(t,\vmax) - v z(t,v) - \int_{\sigma=v}^{\sigma=\vmax} z(t, \sigma) d \sigma. & \text{for all } t\in T \text{ and } v\in V \tag{PI-r-cont}\label{eq:PI-r}\\
        p(t) & = p(0) + \int_{0}^t Q_{t'} (\tau, \tau) d \tau.& \text{for all } t \in T \tag{PI-p-cont}
    \end{align}
\end{definition}

\noindent
The following lemma then establishes IC constraints purely in terms of the allocation function $z$.
\begin{lemma}\label{Lem:PaymentIdentity}
    In the continuous setting, a return policy mechanism $(p,z,r)$ satisfies the IC and IR constraints in \eqref{eq:formulation} if and only if:
    \begin{itemize}
        \item $z(t,v)$ is weakly increasing in $v$ for all $t\in T$. 
        \item The payments $p$ and refunds $r$ satisfy the payment identities \eqref{eq:PI-p} and \eqref{eq:PI-r} respectively.
        \item The function $Q$ as defined in \eqref{eq:Def-Q} satisfies for all $t$ and $t'$:
        $Q(t,t) + Q(t',t') \ge Q(t,t') + Q(t', t)$
        \item $p(0)\le Q(0,0)$.
    \end{itemize}
\end{lemma}

\begin{proof}
Let us start by proving the payment identities. The second stage IC constraint in \eqref{eq:formulation} implies that for all $v$ and $t$, we have
\begin{align*}
    v \in \argmax_{v'} \{r(t,v')+ v z(t,v')\}.
\end{align*}
Therefore, the first order condition with respect to $v'$ implies that
\begin{align*}
    r_v(t,v) + vz_v(t,v) =0.
\end{align*}
Taking the integral of the above expression results in 
\begin{align}\label{eq:pf:Lemma:paymentidentity:secondstage}
    r(t,v)= &r(t,\vmax) + \int_{v}^{\vmax} \sigma z_v(t, \sigma) d \sigma \nonumber \\
    = & r(t,\vmax)+\vmax z(t,\vmax) - v z(t,v) - \int_{\sigma=v}^{\vmax} z(t, \sigma) d \sigma,
\end{align}
where we used integration by parts. 
The above equation is a necessary condition for the second phase IC constraint to hold. 
Also, notice that by fixing $t$ and writing the second phase IC constraint once for $v$ and $v'$ and once for $v'$ and $v$ and then taking the summation, we obtain
\begin{align*}
    \left( z(t,v) - z(t,v')\right) \left( v - v'\right) \ge 0.
\end{align*}
Therefore, we must have that, for any given $t$, $z(t, v)$ is increasing in $v$. 

We next establish that when $z(t, v)$ is increasing in $v$, then \eqref{eq:pf:Lemma:paymentidentity:secondstage} is a sufficient condition for the IC constraint. Notice that plugging \eqref{eq:pf:Lemma:paymentidentity:secondstage} into the IC constraint leads to 
\begin{align*}
  & r(t,\vmax) + \vmax z(t,\vmax) - \int_{\sigma=v}^{\vmax} z(t, \sigma) d \sigma  \\
  & \ge   r(t,\vmax) + \vmax z(t,\vmax) - v' z(t,v')- \int_{\sigma=v'}^{\vmax} z(t, \sigma)d \sigma+ v z(t, v').
\end{align*}
This is equivalent to having
\begin{align*}
    \int_{\sigma=v'}^{v} z(t,\sigma) d \sigma \ge z(t, v') \left( v - v'\right).
\end{align*}
This inequality in turn follows from the fact that $z$ is weakly increasing in $v$. This completes the proof of the payment identity for $r$ as well as the first bullet point in the lemma. 

We now proceed to the payment identity for $p$. The first stage IC constraint implies that for all $t$, we have 
\begin{align*}
    t \in \argmax_{t'} -p(t') + Q(t, t').
\end{align*}
Therefore, the first order condition implies that
\begin{align*}
    -p'(t) + Q_{t'}(t,t) =0.
\end{align*}
Taking integral of the above expression results in 
\begin{align}\label{eq:pf:Lemma:paymentidentity:firststage}
    p(t) = p(0) + \int_{0}^t Q_{t'} (\tau, \tau) d \tau.
\end{align}
Moreover, by writing the first stage IC constraint once for $t$ and $t'$ and once for $t'$ and $t$ and taking the summation, we obtain 
\begin{align*}
    Q(t, t) + Q(t', t') \ge Q(t, t') + Q(t', t).
\end{align*}
The fact that the above inequality is a sufficient condition follows by invoking Rochet's theorem (see, e.g., \cite{ashlagi2010monotonicity}).

Finally, we need to argue that it suffices to apply the IR constraint at $t=0$. First, we note that we can write:
\[Q(t,t) = Q(0,0)+\int_0^t Q_t(\tau,\tau)\,d\tau + \int_0^t Q_{t'}(\tau,\tau)\,d\tau\]
where $Q_t$ and $Q_{t'}$ are the derivatives of $Q$ in its first and second arguments respectively. Then, using \eqref{eq:pf:Lemma:paymentidentity:firststage}, we get:
\begin{align*}
    Q(t,t)-p(t) & = Q(0,0)+\int_0^t Q_t(\tau,\tau)\,d\tau + \int_0^t Q_{t'}(\tau,\tau)\,d\tau - p(0)- \int_0^t Q_{t'}(\tau,\tau)\,d\tau \\
    & = Q(0,0) - p(0) +\int_0^t Q_{t}(\tau,\tau)\,d\tau
\end{align*}
Finally, we note that $Q_t(t,t')\ge 0$ for all $t$ and $t'$ because values are increasing functions of $t$. Then, $Q(0,0) - p(0)\ge 0$ implies $Q(t,t) - p(t)\ge 0$ for all $t$.
This completes the proof of the lemma
\end{proof}

\subsection{A virtual value characterization for the continuous setting}

In this subsection we will prove Lemma~\ref{thm:charac:wo:IC} for the continuous setting.
Let us start by simplifying the function $Q(t, t')=\mathbb{E}_{v\sim G_t}\left[\max_{v'}\left\{r(t',v')+v z(t',v')\right\}\right]$. Let us define the mapping $(t,t', v) \mapsto v^*(t,t',v)$ such that 
\begin{align*}
    v^*(t,t',v) \in \argmax_{v'}\{r(t',v')+v z(t',v')\}.
\end{align*}
If there are multiple $v^*$ for which the equality holds, we just pick one of them arbitrarily. Using this notation, we can rewrite the function $Q(t, t')$ as 
\begin{align*}
    Q(t, t') = \mathbb{E}_{v\sim G_t}\left[r(t',v^*(t,t',v))+v z(t',v^*(t,t',v))\right]. 
\end{align*}
As a result, the envelop theorem implies that the partial derivative of $Q(t, t')$ with respect to $t'$ is given by
\begin{align*}
    Q_{t'}(t, t')= \mathbb{E}_{v\sim G_t}\left[r_t(t', v^*(t,t',v))+v z_t(t', v^*(t,t',v)) \right]. 
\end{align*}
We now proceed with the proof of the theorem. When the two IC constraints hold, Lemma \ref{Lem:PaymentIdentity} enables us to write down the auctioneer's objective in terms of a single function $z(\cdot, \cdot)$. In particular, by invoking the payment identities of Lemma \ref{Lem:PaymentIdentity} and the above expression, and recalling that $v^*(t,t,v) = v$ for all $t,v$, the platform's objective becomes
\begin{align*}
    \mathbb{E}_t\left[ 
    p(t)\right] - \mathbb{E}_{t,v}\left[r(t, v)\right] & =    \mathbb{E}_t\left[ 
     p(0) + \int_{0}^t Q_{t'} (\tau, \tau) d \tau\right]  - \mathbb{E}_{t,v}\left[r(t, v)\right] \\
     & = p(0) + \mathbb{E}_t\left[ 
     \int_{0}^t \mathbb{E}_{y\sim G_\tau} \left[ r_t(\tau, y) + y z_t(\tau, y) \right] d \tau\right]  - \mathbb{E}_{t,v}\left[r(t, v)\right]
\end{align*}
To simplify this expression further, we take derivatives with respect to $t$ of both sides of the payment identity~\eqref{eq:PI-r}:
\[ r_t(t,y) = r_t(t,\vmax)+\vmax z_t(t, \vmax) - y z_t(t,y) - \int_{\sigma=y}^{\sigma=\vmax} z_t(t,\sigma) d\sigma\]
Continuing our sequence of equations, we get:
\begin{align*}
    \mathbb{E}_{t,v}\left[ 
    p(t)-r(t, v)\right] & = p(0) + \mathbb{E}_{t}\left[ \int_0^t r_t(\tau,\vmax) +\vmax z_t(\tau, \vmax)d\tau\right] - \mathbb{E}_{t}\left[\int_0^t\mathbb{E}_{y\sim G_\tau}\left[\int_y^{\vmax} z_t(\tau,\sigma) d\sigma\right] d\tau \right] - \mathbb{E}_{t,v}\left[r(t, v)\right]
\end{align*}
Substituting the payment identity ~\eqref{eq:PI-r} yet again in this equation cancels out the $r(t,\vmax)+\vmax z(t,\vmax)$ terms and we get:
\begin{align*}
= p(0) -r(0,\vmax)- \vmax z(0,\vmax) - \mathbb{E}_{t}\left[\int_0^t\mathbb{E}_{y\sim G_\tau}\left[\int_y^{\vmax} z_t(\tau,\sigma) d\sigma\right] d\tau \right] + \mathbb{E}_{t,v}\left[v z(t,v) + \int_{\sigma=v}^{\sigma=\vmax} z(t,\sigma) d\sigma \right]
\end{align*}
We now further simplify the two integrals over $\sigma$.
\begin{align*}
\mathbb{E}_{t,v}\left[\int_{\sigma=v}^{\sigma=\vmax} z(t,\sigma)d\sigma\right] & = \mathbb{E}_t\left[\int_{v=\vmin}^{v=\vmax}\int_{\sigma=v}^{\sigma=\vmax} z(t,\sigma)d\sigma g_t(v)dv  \right]
= \mathbb{E}_t\left[\int_{\sigma=\vmin}^{\sigma=\vmax}\int_{v=\vmin}^{v=\sigma}g_t(v)z(t,\sigma)dv d\sigma  \right]\\
& = \mathbb{E}_t\left[ \int_{\sigma=\vmin}^{\sigma=\vmax} G_t(\sigma) z(t,\sigma)d\sigma\right]
= \mathbb{E}_{t,v}\left[\frac{G_t(v)z(t,v)}{g_t(v)} \right]
\end{align*}
Finally, we tackle the integral involving $z_t$. 
\begin{align*}
\mathbb{E}_{t}\left[\int_0^t\mathbb{E}_{y\sim G_\tau}\left[\int_y^{\vmax} z_t(\tau,\sigma) d\sigma\right] d\tau \right] & = \int_{t=0}^{t=1} \int_{\tau=0}^{\tau=t}\int_{y=\vmin}^{y=\vmax}\int_{\sigma=y}^{\sigma=\vmax} z_t(\tau,\sigma) g_\tau(y) d\sigma dy d\tau dt\\ 
& = \int_{\tau=0}^{\tau=1}\int_{y=\vmin}^{y=\vmax}\int_{\sigma=y}^{\sigma=\vmax} \int_{t=\tau}^{t=1}z_t(\tau,\sigma) g_\tau(y) dt d\sigma dy d\tau \\
& = \int_{\tau=0}^{\tau=1} \int_{y=\vmin}^{y=\vmax}\int_{\sigma=y}^{\sigma=\vmax} (1-\tau)z_t(\tau,\sigma) g_\tau(y) d\sigma dy d\tau\\
& = \int_{\sigma=\vmin}^{\sigma=\vmax} \int_{\tau=0}^{\tau=1}\int_{y=\vmin}^{y=\sigma} (1-\tau)z_t(\tau,\sigma) g_\tau(y) dy d\tau d\sigma \\
& = \int_{\sigma=\vmin}^{\sigma=\vmax} \int_{\tau=0}^{\tau=1} (1-\tau)z_t(\tau,\sigma) G_\tau(\sigma) d\tau d\sigma \\
\end{align*}
We will now integrate this final integrand by parts. 
\begin{align*}
    \int_{\tau=0}^{\tau=1} (1-\tau)z_t(\tau,\sigma) G_\tau(\sigma) d\tau & = (1-\tau)z(\tau,\sigma) G_\tau(\sigma) |_0^1 - \int_{\tau=0}^{\tau=1} z(\tau,\sigma) \left(-G_\tau(\sigma) + (1-\tau) \frac{dG_\tau(\sigma)}{dt}\right) d\tau \\
    & = -z(0,\sigma)G_0(\sigma) + \int_{\tau=0}^{\tau=1} z(\tau,\sigma) \left(G_\tau(\sigma) - (1-\tau) \frac{dG_\tau(\sigma)}{dt}\right) d\tau 
\end{align*}
Finally, substituting these expressions back into the expression for the revenue and renaming variables, we get,
\begin{align*}
& \mathbb{E}_{t,v}\left[p(t)-r(t, v)\right] & \\
& = p(0) -r(0,\vmax)- \vmax z(0,\vmax) 
+ \mathbb{E}_{v\sim G_0}\left[z(0,v)\frac{G_0(v)}{g_0(v)}\right] & & - \mathbb{E}_{t,v}\left[ z(t,v) \left(\frac{G_t(v)}{g_t(v)} - \frac{1-t}{g_t(v)} \frac{dG_t(v)}{dt}\right) \right]\\
& & & + \mathbb{E}_{t,v}\left[v z(t,v) + \frac{G_t(v)z(t,v)}{g_t(v)}\right] \\
& = p(0) -r(0,\vmax)- \vmax z(0,\vmax) 
+ \mathbb{E}_{v\sim G_0}\left[z(0,v)\frac{G_0(v)}{g_0(v)}\right] & & + \mathbb{E}_{t,v}\left[ z(t,v) \left\{ v + \frac{1-t}{g_t(v)} \frac{dG_t(v)}{dt} \right\} \right]\\
& = p(0) -r(0,\vmax)- \vmax z(0,\vmax) 
+ \mathbb{E}_{v\sim G_0}\left[z(0,v)\frac{G_0(v)}{g_0(v)}\right] & & + \mathbb{E}_{t,v}\left[ z(t,v) \phi(t,v)\right]
\end{align*}

\noindent
In order to eliminate the $t=0$ terms in the above expression, we consider the utility of the buyer at type $0$.
Writing the utility as a function of type $t$ and applying the payment identity, we get,
\begin{align*}
     & -p(t)+Q(t,t) = -p(t) + \mathbb{E}_{v\sim G_t}\left[ r(t,v) + v z(t, v)\right] =  - p(0) - \int_{0}^t Q_{t'} (\tau, \tau) d \tau + \mathbb{E}_{v\sim G_t}\left[ r(t,v) + v z(t, v)\right] 
\end{align*}
Performing the same simplifications as above, we obtain
\begin{align*}
    & -p(t)+Q(t,t) = - p(0) + r(0,\vmax) + \vmax z(0,\vmax) 
- \mathbb{E}_{v\sim G_0}\left[z(0,v)\frac{G_0(v)}{g_0(v)}\right] - \int_{\sigma=\vmin}^{\sigma=\vmax} \int_{\tau=0}^{\tau=t} z(t,\sigma) \frac{dG_t(\sigma)}{dt} d\tau d\sigma  
\end{align*}
When $t=0$, the last term is $0$. We can therefore substitute the equation back into the expected revenue equation to obtain:
\begin{align*}
\mathbb{E}_{t,v}\left[p(t)-r(t, v)\right] = p(0)-Q(0,0) + \mathbb{E}_{t,v}\left[ z(t,v) \phi(t,v)\right]
\end{align*}
completing the proof of the lemma.
\hfill $\square$

\subsection{Payment identities and virtual values for the discrete setting}

Next we consider the discrete setting and prove a counterpart of Lemma~\ref{Lem:PaymentIdentity} and Lemma~\ref{thm:charac:wo:IC} for the discrete setting. Recall that in the discrete setting, there is a finite number of discrete buyer types and possible values. We assume buyer's type set and value set being denoted by $T=\{0=t_1<t_2<\cdots<t_m=1\}$ and $V=\{v_1<v_2<\cdots<v_n\}$. Each buyer type appears with probability $1/m$. The discrete case looks very similar to the continuous case. We will now expand on the details. 

We first define a payment identity and virtual values.

\begin{definition}[Discrete payment identity and virtual values]\label{Def:DiscPaymentIdentity}
    For any allocation function $z$, the discrete payment identities for payments $p$ and refunds $r$ are defined as:
    \begin{align}
        r(t_i,v_j) &= r(t_i,v_1)+z(t_i,v_1)\cdot v_1 - z(t_i,v_j)\cdot v_j +\sum_{x=2}^{j}z(t_i,v_x)(v_x-v_{x-1}) & \forall  i\in [m], j\in [n] \label{eq:PI-r-disc}\tag{PI-r-disc}\\
        p(t_i) & =p(t_1)+\sum_{x=2}^{i} (Q(t_{x},t_{x})-Q(t_{x},t_{x-1}))& \forall i \in [m] 
        \tag{PI-p-disc}
    \end{align}
We define the buyer's virtual value function as
    \begin{align*}
        \phi(t_i,v_j) = v_j -(v_j-v_{j-1})\cdot (m-i) \cdot \frac{G_{t_i}(v_{j-1})-G_{t_{i+1}}(v_{j-1})}{g_{t_i}(v_j)} & \,\forall i\in [m], j\in [n] \tag{VV-disc}
    \end{align*}
\end{definition}
Note the similarity between the discrete virtual value and the continuous version in Definition~\ref{def:virtual:main}. Here $\frac{1-F(t_i)}{f(t_i)}=(m-i)$; the discrete derivative of the distribution function $G_{t_i}(v_j)$ along the value dimension is $g_{t_i}(v_j)/(v_j-v_{j-1})$, and the role of the discrete derivative of $G$ along the type dimension is played by $-(G_{t_i}(v_{j-1})-G_{t_{i+1}}(v_{j-1}))$.

Armed with these definitions, we have the following counterparts of the payment identities and virtual value characterization. 
\begin{lemma}\label{Lem:PaymentIdentity-disc}
    In the discrete setting, a return policy mechanism $(p,z,r)$ satisfies the IC and IR constraints in \eqref{eq:formulation} if and only if:
    \begin{itemize}
        \item $z(t,v)$ is weakly increasing in $v$ for all $t\in T$. 
        \item The payments $p$ and refunds $r$ satisfy the payment identities \eqref{eq:PI-p-disc} and \eqref{eq:PI-r-disc} respectively.
        \item The function $Q$ as defined in \eqref{eq:Def-Q} satisfies for all $t$ and $t'$:
        $Q(t,t) + Q(t',t') \ge Q(t,t') + Q(t', t)$
        \item $p(0)\le Q(0,0)$.
    \end{itemize}
\end{lemma}

\begin{lemma}\label{Lem:vv-char-disc}
    The expected revenue of any BIC mechanism $(p, z, r)$ in the discrete setting is given by:
    \begin{align*}
        \mathbb{E}_{t,v}[p(t)-r(t,v)] = p(t_1)-Q(t_1, t_1)+\mathbb{E}_{t,v}[\phi(t,v)z(t,v)]
    \end{align*}
    where $p(t_1)-Q(t_1, t_1)$ is the utility of a buyer with type $t_1$ and is, by default, assumed to be $0$.
\end{lemma}
The seller's optimization problem in the discrete setting is identical to that in program \eqref{eq:formulation}.

\vspace{0.1in}
\noindent
{\em Proof of Lemma~\ref{Lem:PaymentIdentity-disc}.}
The first, third and fourth bullet point can be proved with similar proofs from Lemma~\ref{Lem:PaymentIdentity}, so we will only prove the payment identity, which is the second bullet point.

We first prove the payment identity for $r$. The second stage IC constraint in \eqref{eq:formulation} implies that for all $t_i$ and $v_j$, we have
\begin{align*}
    v_j \in \argmax_{v} \{r(t_i,v)+ v_j z(t_i,v)\}.
\end{align*}
Therefore, this implies that for any $t_i$ and $v_j$
\begin{align*}
    r(t_i,v_j) - r(t_i,v_{j-1}) =  -z(t_i,v_j)(v_j-v_{j-1}).
\end{align*}
Taking the sum from of the above expression results in 
\begin{align}\label{eq:pf:Lemma:paymentidentity-disc:secondstage}
    r(t_i,v_j) &= r(t_i,v_1)+z(t_i,v_1)\cdot v_1 - z(t_i,v_j)\cdot v_j +\sum_{x=2}^{j}z(t_i,v_x)(v_x-v_{x-1}) & \forall  i\in [m], j\in [n] 
\end{align}
which proves the payment identity for $r$. 

Now we proceed to the payment identity for $p$. The first stage IC constraint implies that for all $t_i$, we have 
\begin{align*}
    t_i \in \argmax_{t} -p(t) + Q(t_i, t).
\end{align*}
Therefore, this implies that for any $t_i$
\begin{align*}
    p(t_i)-p(t_{i-1}) = Q(t_i,t_i)-Q(t_i,t_{i-1}).
\end{align*}
Taking the sum of the above expression from $2$ to $i$ results in 
\begin{align}\label{eq:pf:Lemma:paymentidentity-disc:firststage}
    p(t_i) & =p(t_1)+\sum_{x=2}^{i} (Q(t_{x},t_{x})-Q(t_{x},t_{x-1}))& \forall i \in [m]
\end{align}
which proves the lemma.

\hfill$\square$

\vspace{0.1in}
\noindent
{\em Proof of Lemma~\ref{Lem:vv-char-disc}.}

Like what we've done for the proof of Lemma~\ref{thm:charac:wo:IC}, we will first derive a simpler expression for $Q(t_i,t_j)$. Recall that $v^*(t,t',v) = v$ for all $t,t',v$, we have
\begin{align*}
    Q(t_x, t_{x'}) = \sum_{y=1}^{n}\left[(r(t_{x'},v_y)+z(t_{x'},v_y)v_y)g_{t_x}(v_y)\right]. 
\end{align*}

We can still write down the auctioneer's objective in terms of a single function $z(\cdot, \cdot)$. In particular, by invoking the payment identities of Lemma \ref{Lem:PaymentIdentity-disc}, the platform's objective becomes
\begin{align*}
    \mathbb{E}_t\left[ 
    p(t)\right] - \mathbb{E}_{t,v}\left[r(t, v)\right] & =    \frac{1}{m}\sum_{i=1}^{m}p(t_i)-\frac{1}{m}\sum_{i=1}^{m}\sum_{j=1}^{n}r(t_i,v_j)g_{t_i}(v_j) \\
     & = p(t_1) + \frac{1}{m}\sum_{i=2}^{m}\sum_{x=2}^{i}\sum_{y=1}^{n}\left[((r(t_x,v_y)+z(t_x,v_y)v_y)-(r(t_{x-1},v_y)+z(t_{x-1},v_y)v_y))g_{t_x}(v_y) \right]\\
     & \ \  \     -\frac{1}{m}\sum_{i=1}^{m}\sum_{j=1}^{n}(r(t_i,v_1)+z(t_i,v_1)\cdot v_1 - z(t_i,v_j)\cdot v_j +\sum_{x=2}^{j}z(t_i,v_x)(v_x-v_{x-1}))g_{t_i}(v_j)
\end{align*}

We will tackle the two sums by exchanging the order of sums. Consider the first part:
\begin{align*}
    &\sum_{i=2}^{m}\sum_{x=2}^{i}\sum_{y=1}^{n}\left[((r(t_x,v_y)+z(t_x,v_y)v_y)-(r(t_{x-1},v_y)+z(t_{x-1},v_y)v_y))g_{t_x}(v_y) \right]\\
    & = \sum_{i=2}^{m}\sum_{x=2}^{i}\sum_{y=1}^{n}(((r(t_x,v_1)+z(t_x,v_1)v_1 +\sum_{l=2}^{y}z(t_x,v_l)(v_l-v_{l-1}))\\
    & -(r(t_{x-1},v_1)+z(t_{x-1},v_1)v_1 +\sum_{l=2}^{y}z(t_{x-1},v_l)(v_l-v_{l-1})))g_{t_x}(v_y))\\
    & = \sum_{i=1}^{m}(r(t_i,v_1)+z(t_i,v_1)v_1)-m(r(t_1,v_1)+z(t_1,v_1)v_1)+\sum_{i=2}^{m}\sum_{x=2}^{i}\sum_{y=2}^{n}\sum_{l=2}^{y}(z(t_x,v_l)-z(t_{x-1},v_l))(v_l-v_{l-1})g_{t_x}(v_y)\\
    & = \sum_{i=1}^{m}(r(t_i,v_1)+z(t_i,v_1)v_1)-m(r(t_1,v_1)+z(t_1,v_1)v_1)+\sum_{i=2}^{m}\sum_{j=2}^{n}(m-i+1)(z(t_i,v_j)-z(t_{i-1},v_j))(v_j-v_{j-1})(1-G_{t_i}(v_{j-1}))
\end{align*}

For the second sum, we also apply the trick of exchanging the order of sums,
\begin{align*}
    &\sum_{i=1}^{m}\sum_{j=1}^{n}(r(t_i,v_1)+z(t_i,v_1)\cdot v_1 - z(t_i,v_j)\cdot v_j +\sum_{x=2}^{j}z(t_i,v_x)(v_x-v_{x-1}))g_{t_i}(v_j)\\
    & = \sum_{i=1}^{m}(r(t_i,v_1)+z(t_i,v_1)v_1)-\sum_{i=1}^{m}\sum_{j=1}^{n}z(t_i,v_j) v_jg_{t_i}(v_j)+\sum_{i=1}^{m}\sum_{j=2}^{n}z(t_i,v_j)(v_j-v_{j-1})(1-G_{t_i}(v_{j-1}))
\end{align*}

Finally, substituting these expressions back into the original expression for the revenue and renaming variables, we get,
\begin{align*}
    \mathbb{E}_{t,v}[p(t)-r(t,v)]&=p(t_1)-Q(t_1,t_1)+\sum_{i=1}^{m}\sum_{j=1}^{n}z(t_i,v_j)(v_jg_{t_i}(v_j) -(v_j-v_{j-1})(m-i) (G_{t_i}(v_{j-1})-G_{t_{i+1}}(v_{j-1})))\\
    &=p(t_1)-Q(t_1,t_1)+\mathbb{E}_{t,v}\left[ z(t,v) \phi(t,v)\right]
\end{align*}
completing the proof of the lemma.
\hfill$\square$

\end{document}